\documentclass[letterpaper, 10 pt, conference]{ieeeconf}  %

\IEEEoverridecommandlockouts                              %
\usepackage{amsmath}
\usepackage{amssymb}
\usepackage{graphicx}
\usepackage{cleveref}
\usepackage{mleftright}
\usepackage{cases}
\usepackage{empheq}
\usepackage{bm}
\usepackage{bbm}
\usepackage{mathtools}
\usepackage{nicefrac}
\usepackage{listings}
\usepackage[colorinlistoftodos]{todonotes}
\usepackage{caption}
\usepackage{subcaption}
\usepackage{booktabs}
\usepackage{xspace}
\usepackage{microtype}
\usepackage{siunitx}
\usepackage{cite}
\usepackage{tikz}

\usepackage{dsfont}

\newtheorem{theorem}{Theorem}

\newtheorem{corollary}{Corollary}
\newtheorem{assumption}{Assumption}
\newtheorem{problem}{Problem}

\newtheorem{definition}{Definition}

\crefname{equation}{Eq.}{Eqs.}
\crefname{pluralequation}{Eqs.}{Eqs.}

\crefname{algorithm}{Algorithm}{Algorithm}

\crefname{figure}{Fig.}{Figs.}
\crefname{pluralfigure}{Figs.}{Figs.}

\crefname{section}{Sect.}{Sects.}
\crefname{pluralsection}{Sects.}{Sects.}

\crefname{table}{Table}{Table}
\crefname{pluraltable}{Tables}{Tables}

\crefname{definition}{Def.}{Def.}
\crefname{pluraldefinition}{Defs.}{Defs.}

\crefname{theorem}{Theorem}{Theorems}
\crefname{pluraltheorem}{Theorems}{Theorems}

\crefname{lemma}{Lemma}{Lemmas}
\crefname{plurallemma}{Lemmas}{Lemmas}

\crefname{example}{Example}{Example}
\crefname{pluralexample}{Examples}{Examples}

\crefname{problem}{Problem}{Problem}
\crefname{pluralproblem}{Problems}{Problems}

\crefname{assumption}{Assumption}{Assumption}
\crefname{pluralassumption}{Assumptions}{Assumptions}

\crefname{remark}{Remark}{Remark}
\crefname{pluralremark}{Remarks}{Remarks}

\crefname{proposition}{Proposition}{Proposition}
\crefname{pluralproposition}{Propositions}{Propositions}

\crefname{corollary}{Corollary}{Corollary}
\crefname{pluralcorollary}{Corollaries}{Corollaries}

\crefname{appendix}{Appendix}{Appendices}
\crefname{pluralappendix}{Appendices}{Appendices}

\newcommand*{\NN}{\mathbb{N}}

\newcommand*{\RR}{\mathbb{R}}

\newcommand{\system}{\mathcal{S}}
\newcommand*{\policy}{\ensuremath{\pi}}

\newcommand{\satprob}{\mathrm{Pr}}
\newcommand{\Prob}{\mathbb{P}}

\newcommand*{\distr}[1]{\Delta(#1)}

\newcommand{\tuple}[1]{\ensuremath{( #1 )}}

\newcommand*{\States}{\ensuremath{S}}
\newcommand*{\Actions}{\ensuremath{Act}}
\newcommand*{\initState}{\ensuremath{s_I}}

\newcommand*{\transprob}{\ensuremath{\psi}}
\newcommand*{\transfuncImdp}{\ensuremath{\mathcal{P}}}
\newcommand*{\labelfunc}{\ensuremath{L}}
\newcommand*{\Labels}{\ensuremath{\mathcal{Y}}}
\newcommand*{\IMDP}{\ensuremath{\tuple{\States,\initState,\Actions,\transfuncImdp,\Labels,\labelfunc}}}
\newcommand*{\imdp}{{\ensuremath{\mathcal{M}_\Interval}}}

\newcommand*{\scheduler}{\ensuremath{\sigma}}
\newcommand*{\Scheduler}{\ensuremath{\mathfrak{S}}}
\newcommand*{\schedulerSpace}{\ensuremath{\Scheduler}}

\newcommand*{\Interval}{\ensuremath{\mathbb{I}}}
\newcommand*{\plow}{\ensuremath{\check{p}}}
\newcommand*{\pupp}{\ensuremath{\hat{p}}}

\newcommand*{\horizon}{\ensuremath{h}}

\DeclareMathOperator*{\argmax}{arg\,max}

\newcommand{\noise}{{\eta}}

\newcommand{\Bin}{B}

\title{\LARGE \bf
Data-Driven Abstraction and Synthesis for \\ Stochastic Systems with Unknown Dynamics
}
\author{Mahdi Nazeri, Thom Badings, Anne-Kathrin Schmuck, Sadegh Soudjani, and Alessandro Abate%
\thanks{M. Nazeri, T. Badings, and A. Abate are with the Department of Computer Science, University of Oxford, Oxford, United Kingdom. {\tt\small \{thom.badings,alessandro.abate\}} {\tt\small @cs.ox.ac.uk}}%
\thanks{M. Nazeri, A. Schmuck, S. Soudjani are with the Max Planck Institute for Software Systems, Kaiserslautern, Germany. {\tt\small \{mnazeri,akschmuck,sadegh\}@mpi-sws.org}}%
\thanks{M. Nazeri and A. Schmuck are supported by the DFG  grants SCHM 3541/1-1 and 389792660 TRR 248 (CPEC). T. Badings and A. Abate are supported by EPSRC grant EP/Y028872/1. %
The research of S. Soudjani is supported by EIC 101070802 and ERC
101089047.
}
}

\begin{document}

\maketitle
\thispagestyle{empty}
\pagestyle{empty}

\begin{abstract}%
We study the automated abstraction-based synthesis of correct-by-construction control policies for stochastic dynamical systems with unknown dynamics. Our approach is to learn an abstraction from sampled data, which is represented in the form of a finite Markov decision process (MDP). In this paper, we present a data-driven technique for constructing finite-state interval MDP (IMDP) abstractions of stochastic systems with unknown nonlinear dynamics. As a distinguishing and novel feature, our technique only requires (1) noisy state-input-state observations and (2) an upper bound on the system's Lipschitz constant. Combined with standard model-checking techniques, our IMDP abstractions enable the synthesis of policies that satisfy probabilistic temporal properties (such as ``reach-while-avoid'') with a predefined confidence. Our experimental results show the effectiveness and robustness of our approach.
\end{abstract}%

\section{Introduction}
\label{sec:introduction}

Ensuring safe control of safety-critical systems, such as power systems, autonomous vehicles, and air traffic control, is of paramount importance. 
Often, these systems can be effectively modelled as stochastic dynamical systems. 
The problem of designing a control policy that provably satisfies a given specification is referred to as \emph{formal policy synthesis}~\cite{belta2017formal}. A simple yet ubiquitous example of such a specification is the \emph{reach-while-avoid task}~\cite{DBLP:conf/cav/FanMM018,DBLP:journals/automatica/SummersL10}. A reach-while-avoid task requires the policy to drive the system to a set of goal states within a fixed time horizon, while always avoiding unsafe states.
For stochastic systems, the goal is to synthesise a policy that maximises the probability that the closed-loop system satisfies the reach-while-avoid specification.

A classical approach to formal policy synthesis is to construct a finite-state abstraction of the stochastic dynamical system and synthesise an optimal policy over the abstract model~\cite{LSAZ21,DBLP:journals/automatica/AbatePLS08,DBLP:books/daglib/0032856}.
However, these abstraction techniques are largely \emph{model-based} and thus require the system's dynamics to be known.
In many scenarios, the dynamics are uncertain or unknown at all, rendering model-based abstraction techniques inapplicable.
Recently, \emph{data-driven abstractions} have been proposed to construct sound abstractions from observations~\cite{DBLP:conf/adhs/MakdesiGF21,DBLP:journals/csysl/CoppolaPM23,lavaei2023compositional,DBLP:journals/corr/abs-2206-08069,DBLP:journals/automatica/HashimotoSKUD22,DBLP:conf/cdc/DevonportSA21,DBLP:journals/corr/abs-2303-17618,DBLP:journals/csysl/PeruffoM23,schon2024data,DBLP:conf/l4dc/GraciaBLL24,DBLP:conf/hybrid/JacksonLFL21,DBLP:journals/jair/BadingsRAPPSJ23,DBLP:journals/csysl/LavaeiSFZ23, ajeleye2023data}. These methods rely on inferring information from data to compensate for the lack of complete knowledge about the dynamics. 
However, aside from a few recent exceptions~\cite{l4dc,DBLP:conf/l4dc/GraciaBLL24,DBLP:journals/corr/abs-2412-11343,DBLP:journals/csysl/LavaeiSFZ23}, these approaches are limited to non-probabilistic and/or linear systems. Other approaches from related areas, such as reinforcement learning, can generalise to different settings but lack rigorous guarantees.

In this work, we introduce a novel data-driven abstraction algorithm for discrete-time, nonlinear stochastic systems. 
We represent abstractions as interval Markov decision processes (IMDPs) and assume no knowledge of either the system dynamics or the probability distribution of the disturbances. 
Our method requires only (1) noisy state-input-state observations and (2) a bound on the Lipschitz constant of~the~dynamics.

Contrary to assumption (1), existing approaches make the (often unrealistic) assumption that the stochastic and non-stochastic part of the dynamics can be sampled individually~\cite{DBLP:journals/corr/abs-2412-11343,l4dc}. The noisy state-input-state data used by our method can be easily obtained by directly sampling the (noisy) system behaviour, either in open- or closed-loop operation. 

Regarding requirement (2), Lipschitz continuity is a standard assumption in control theory and plays a central role in our approach. 
Lipschitz constants can be computed directly when the dynamics are known and otherwise estimated through various approaches depending on the available information about the dynamics.
When a prior distribution over the dynamics is known, statistical techniques~\cite{casella2002statistical} can be used to compute confidence intervals on the Lipschitz constant based on independent parameter samples. 
If only observations are available, data-driven methods enable robust estimation of Lipschitz constants from noisy or noise-free samples under additional smoothness assumptions~\cite{zhang2024formal,DBLP:journals/tmlr/HuangRC23}.

Closely related works are~\cite{l4dc, DBLP:journals/corr/abs-2412-11343}; however, contrary to our setting, these approaches require noise-free samples of the system's dynamics.  \cite{l4dc}~constructs an IMDP abstraction given noise-free observations and a Lipschitz constant of the system, and requires the noise to be additive.  The work\cite{DBLP:journals/corr/abs-2412-11343} assumes only the disturbance is unknown and builds abstractions as (more general) robust MDPs. Finally, the paper \cite{DBLP:journals/csysl/LavaeiSFZ23}~constructs MDP/IMDP abstractions from noisy trajectories but can only provide guarantees for finite-horizon specifications. Similar guarantees can alternatively be provided via certificate synthesis, as pursued for a model-free setup in \cite{rickard2025data}.

In summary, we present a novel data-driven algorithm for constructing IMDP abstractions of discrete-time stochastic systems with unknown dynamics. Our abstractions are probably approximately correct (PAC), enabling the synthesis of policies with guarantees. To the best of our knowledge, our approach is the first data-driven abstraction technique for fully unknown stochastic systems that can be used for both finite and infinite horizon specifications. 

\section{Problem Formulation}
A probability space $\tuple{\Omega, \mathcal{F}, \mathbb{P}}$ consists of a sample space $\Omega$, a $\sigma$-algebra $\mathcal{F}$, and a probability measure $\Prob \colon \mathcal{F} \to [0,1]$.
A random variable $\noise$ is a measurable function $\noise \colon \Omega \to \RR^n$, $n \in \NN$, which takes value $\noise(\omega) \in \RR^n$ for $\omega \in \Omega$.
We write $\distr{X}$ for the set of distributions over an (in)finite set $X$.

\label{sec:problem}
\textbf{Stochastic systems.} A stochastic dynamical system $\system$ is defined over discrete steps $k \in \NN$ by the difference equations
\begin{equation}
    x_{k+1} = f(x_k, u_k, \noise_k), \quad 
    y_{k} = g(x_{k}),
    \label{eq:system}
\end{equation}
where $x_k \in \mathcal{X} \subset \mathbb{R}^n$ is the state at step $k$, with $x_0 = x_I$ the initial state and $\mathcal{X}$ a compact state space.
The labelling function $g \colon \mathcal{X} \to 2^{\mathcal{Y}}$ maps each state to a (possibly empty) subset of the finite set of labels $\mathcal{Y}$.
The control input at step $k$ is $u_k \in \mathcal{U} \subseteq \mathbb{R}^m$, and $\{\noise_k\}_{k \in \NN}$ is a sequence of i.i.d. random variables $\noise_k \colon \Omega \to \mathcal{W} \subseteq \RR^p$ defined on the probability space $(\Omega, \mathcal{F}, \mathbb{P})$.
We make the (standard) assumption that $f \colon \mathcal{X} \times \mathcal{U} \times \mathcal{W} \to \mathcal{X}$ is Lipschitz continuous. The labelling function $g$ is fully known, but the function $f$ is unknown.

\begin{definition}[Lipschitz]
    \label{def:Lipschitz}
    System $\system$ is Lipschitz continuous  in $x$ and $u$ if there exist constants \(L_X, L_U \in \mathbb{R}_{>0}\) such that
    \begin{align}
    &{} |f(x_1, u_1, \noise(\omega)) - f(x_2, u_2, \noise(\omega))| \leq L_X|x_1 - x_2|
    \nonumber
    \\
    \label{eq:lipschitz}
    &{} \quad+ L_U|u_1 - u_2| \quad \forall x_1, x_2 \in \mathcal{X}, \,\, u_1, u_2 \in \mathcal{U}, \,\, w \in \Omega,
    \end{align}
    where \(|\cdot|\) is an appropriate norm on the respective spaces.
\end{definition}
    
\begin{assumption}[Unknown dynamics]
    \label{assumption}
    The dynamics function $f$ of the system $\system$ is unknown, but $f$ is assumed to be Lipschitz continuous, and constants $L_X$ and $L_U$ (or valid upper bounds of them) are available.
\end{assumption}

\textbf{Policies.}
A (Markovian) policy for system $\system$ is a function $\policy \coloneqq (\policy_0, \policy_1, \ldots)$, where each $\policy_k \colon \RR^n \to \mathcal{U}$, $k \in \NN$, is a measurable map.
We denote the set of all policies for system $\system$ by $\Pi^\system$.
For a fixed policy $\pi$, the state evolves according to $x_{k+1} = f(x_k, \pi(x_k), \noise_k)$, which induces a Markov process over (infinite) state \emph{trajectories} $\{x_0, x_1, \ldots\}$, and thus also over the \emph{traces} $\{y_0, y_1, \ldots\}$, with each $y_k \in 2^\mathcal{Y}$ a subset of labels. We denote the induced probability measure associated with this Markov process over traces by $\mathbb{P}_\pi^\system$~\cite{Bertsekas.Shreve78, DBLP:books/wi/Puterman94}.

\textbf{Specifications.}
A specification $\varphi$ for the system $\system$ is a measurable subset of the set of (infinite) labelling~trajectories,\footnote{We use $A^\mathbb{N}$ to denote the set of all infinite strings of the form $(a_0, a_1, a_2, \ldots)$, where $a_k \in A$ takes value from $A$ for every $k\in\mathbb{N}$.} i.e., $\varphi\subseteq{(2^\mathcal{Y})}^\mathbb{N}$, where $\mathbb{N}$ includes $0$.

As an example, consider a system $\mathcal{S}$ as in \eqref{eq:system} with goal states $\mathcal{X}_G \subset \mathcal{X}$, unsafe states $\mathcal{X}_U \subset \mathcal{X}$, and a labelling function $g \colon \mathcal{X} \to 2^\mathcal{Y}$ over $\mathcal{Y} = \{\mathsf{G},\mathsf{U}\}$ s.t.\ for all $x \in \mathcal{X}$,
\[
       x \in \mathcal{X}_G \iff \mathsf{G} \in g(x),
        \quad\text{and}\quad
       x \in \mathcal{X}_U \iff \mathsf{U} \in g(x).
    \]
\noindent
A system with this labelling function supports the commonly used \emph{reach-while-avoid} specification $\varphi_\text{rwa}$:
    \begin{equation*} %
        \varphi_\text{rwa} \coloneqq \{ 
       (y_0, y_1, \ldots) : {}{}\exists k \in \mathbb{N}, \, \mathsf{G} \in y_k \, \wedge
        \forall k'\leq k, \, \mathsf{U} \notin y_{k'}
        \}.
    \end{equation*}
If we also require that $k \leq \horizon$ for some $k \in \NN$, then the specification is \emph{bounded}, i.e., $\mathcal{X}_G$ must be reached within $\horizon$ steps.
Without this extra constraint, the specification has an unbounded horizon, i.e., $\mathcal{X}_G$ must only be \emph{eventually} reached.

\begin{definition}[Satisfaction probability]
    \label{def:satprob_dyn}
    The probability that the system $\system$ with policy $\pi \in \Pi^\system$ satisfies the specification $\varphi$ is $\satprob_\pi^\system(\varphi) \coloneqq \Prob_\pi^\system\{ (y_0, y_1, \ldots) \in \varphi \}$. 
\end{definition}

\textbf{Data-driven synthesis.} 
Given a specification $\varphi$, we aim to synthesise a policy such that $\satprob_\pi^\system(\varphi)$ is at least $\rho \in [0,1]$.
As the dynamics are unknown (cf. \cref{assumption}), we can only use sampled data to find this policy.
Crucially, we assume we can only observe (noisy) \emph{state-input-state} triples $(x_k, u_k, x_k^+)$, but the noise value $\noise_k(\omega)$ that produced $x_k^+$ is unknown.
\begin{definition}[Dataset]
    A dataset for system $\system$ of $N \in \NN$ state-input-state observations is defined as
    \begin{equation*}
    \mathbf{D}_N = \big\{ \big(x_i, u_i, f(x_i,u_i,\noise(\omega_i)) \big) : i \in \{1,\ldots,N\}, \omega_i \in \Omega \big\}.
    \end{equation*}
\end{definition}%
\smallskip
Given the dataset $\mathbf{D}_N$, we aim to solve the next problem.\footnote{The policy $\policy$ depends on the dataset $\mathbf{D}_N$, but we omit this dependency in favor of less cluttered notation.}

\begin{problem}[Data-driven synthesis]
\label{problem}
Given a system $\system$ as in \eqref{eq:system}, Lipschitz constants $L_X$ and $L_U$, a specification $\varphi \subseteq (2^\mathcal{Y})^\mathbb{N}$, a dataset $\mathbf{D}_N = ( x_i, u_i, x_i^+ )_{i=1}^N$, and probabilities $\rho \in [0,1]$ and $\beta \in (0,1)$, compute a policy $\pi \in \Pi^\system$ such that 
\begin{equation*}
  \Prob^N \big\{ (\omega_1, \ldots, \omega_N) \in \Omega^N : 
    \satprob_\pi^\system(\varphi)
     \geq \rho
    \big\} \geq 1 - \beta.   
\end{equation*}
\end{problem}
\smallskip

\cref{problem} requires an algorithm that, given a dataset $\mathbf{D}_N$, synthesises a policy $\pi$ for which the probability of satisfying $\varphi$ is at least $\rho$.
However, since this policy is computed based on a finite dataset, the outer probability shows the chance that we might (with probability at most $\beta)$ have used a dataset leading to an invalid solution, i.e., the synthesised policy satisfies $\varphi$ with probability lower than $\rho$.

\section{Finite-State IMDP Abstraction}
\label{sec:abstraction}
We solve \cref{problem} by learning an abstraction of the system $\system$ from the dataset $\mathbf{D}_N$.
We represent this abstraction as an MDP with intervals of transition probabilities, called an interval MDP (IMDP).
IMDPs have been used for such data-driven abstractions by~\cite{l4dc,DBLP:journals/corr/abs-2412-11343,DBLP:conf/l4dc/GraciaBLL24}, but these approaches require direct access to samples of the noise $\noise(\omega)$, which is unrealistic if the dynamics are unknown.
By contrast, our dataset $\mathbf{D}_N = ( x_i, u_i, x_i^+ )_{i=1}^N$ only contains \emph{noisy observations} and not the actual values of the stochastic noise.

\textbf{IMDPs.}
In this section, we introduce the necessary background on IMDP abstractions, largely following the notation from~\cite{l4dc}.
In \cref{sec:data}, we propose our novel data-driven approach for constructing formally correct abstractions.

\begin{definition}[IMDP]
    \label{def:IMDP}
    An \emph{interval MDP} $\imdp$ is a tuple $\IMDP$, where
    \begin{itemize}
        \item $\States$ is a finite set of states, and $\initState \in S$ is the initial state,
        \item $\Actions$ is a finite set of actions, and we write $\Actions(s) \subseteq \Actions$ for the actions enabled in state $s \in S$,
        \item $\labelfunc \colon \States \to 2^\Labels$ is a labelling function over a finite set of labels $\Labels$ (which is the same as for the system $\system$), and %
        \item $\transfuncImdp \colon \States \times \Actions \rightharpoonup 2^{\distr{\States}}$ is a transition function\footnote{The transition function $\transfuncImdp$ is a partial map, shown with $\rightharpoonup$, due to the fact that not all actions may be enabled in every state.} defined for all $s \in \States, a \in \Actions(s)$ as
    \begin{align*}
        \transfuncImdp(s,a) = \big\{ {}&{} \mu \in \distr{\States} : 
        \forall s' \in \States, \,\,
        \\ 
        &\mu(s') \in [\plow(s,a,s'), \pupp(s,a,s')] \subset [0,1]
        \big\}.
    \end{align*}
     \end{itemize}
\end{definition}

\smallskip
We call $[\plow(s,a,s'), \pupp(s,a,s')] \subseteq [0,1]$ the \emph{probability interval} for the transition $(s,a,s')$.
Actions in an IMDP are chosen by a (Markovian) \emph{scheduler}\footnote{For clarity, in this work we consistently use the word \emph{scheduler} for (finite) IMDPs, and \emph{policy} for (continuous) dynamical systems.} $\scheduler = (\scheduler_0, \scheduler_1, \ldots)$, where each $\scheduler_k \colon \States \to \Actions$ is defined such that $\scheduler_k(s) \in \Actions(s)$ for all $s \in \States$.
The set of all schedulers for $\imdp$ is $\schedulerSpace^\imdp$.

An IMDP defines a game between a scheduler that selects actions and an \emph{adversary} that fixes distributions $P(s,a) \in \transfuncImdp(s,a)$ for all $s \in \States$, $a \in \Actions(s)$.
We assume a different probability can be chosen every time the same pair $(s,a)$ is encountered (called the \emph{dynamic} uncertainty model~\cite{DBLP:journals/mor/Iyengar05}).
We overload notation and write $P \in \transfuncImdp$ for fixing an adversary.

Fixing $\scheduler \in \schedulerSpace^\imdp$ and $P \in \transfuncImdp$ for $\imdp$ yields a Markov chain with (standard) probability measure $\Prob^\imdp_{\scheduler,P}$~\cite{BaierKatoen08}.
Analogous to the dynamical system $\system$, a specification $\varphi'$ for $\imdp$ is a subset of ${(2^\Labels)}^\mathbb{N}$, i.e, a set of (infinite) traces $(\labelfunc(s_0), \labelfunc(s_1), \ldots)$.
The probability that $\imdp$ with scheduler $\scheduler$ and adversary~$P$ satisfies $\varphi'$ is written as $\satprob^\imdp_{\scheduler,P}(\varphi')$.
An optimal (robust) scheduler $\scheduler^\star \in \schedulerSpace^\imdp$ maximises this satisfaction probability:\footnote{Conversely, we may define variants of robust IMDP schedulers where the max. and/or min. operators are flipped.}
\begin{equation}
    \label{eq:optimal_policy}
    \scheduler^\star \in \argmax_{\scheduler \in \schedulerSpace^\imdp} \, \min_{P \in \transfuncImdp} \satprob^\imdp_{\scheduler,P}\{ (L(s_0), L(s_1), \ldots) \in \varphi' \}.
\end{equation}
Specifications for IMDPs are commonly expressed in \emph{linear temporal logic} (LTL) or \emph{probabilistic computation tree logic} (PCTL), for which optimal schedulers can be computed using, e.g., robust value iteration~\cite{DBLP:conf/cdc/WolffTM12, DBLP:journals/mor/Iyengar05}, implemented in mature model checking tools such as PRISM~\cite{DBLP:conf/cav/KwiatkowskaNP11} and Storm~\cite{DBLP:conf/cav/DehnertJK017}.

\textbf{States.} 
We partition the compact state space $\mathcal{X}$ into $v \in \mathbb{N}$ convex polytopic regions $\{R_1, R_2, \dots, R_v\}$, such that
\begin{itemize}
    \item for all $i, j$ with $i \neq j$, it holds that $R_i \cap R_j = \emptyset$;
    \item the partition covers $\mathcal{X}$, i.e., $\bigcup_{i=1}^{v} R_i = \mathcal{X}$.
\end{itemize}
We define an absorbing region, denoted as $R_{\star}$, which contains the rest of the state space, given by $R_{\star} = \mathbb{R}^n \setminus \mathcal{X}$.
We define one IMDP state for each partition region, such that $S \coloneqq \{s_1, s_2, \dots, s_v\} \cup \{s_\star\}$.
The \emph{abstraction map} $\mathcal{T}: \mathbb{R}^n \to S$ assigns to each continuous state an abstract IMDP state, i.e.\ $x \in R_i \iff \mathcal{T}(x) = s_i$.
We also define its preimage $\mathcal{T}^{-1}: S \to 2^{\mathbb{R}^n}$, which maps each IMDP state to its partition region.
The initial IMDP state is $s_I = \mathcal{T}(x_I)$. 

\textbf{Labelling.}
The set of IMDP labels $\Labels$ is the same as for the dynamical system $\system$. We call an IMDP \emph{label-preserving} if for all $x \in \mathcal{X}$, $\labelfunc(\mathcal{T}(x)) = g(x)$.
That is, the label $g(x)$ in state $x$ in system $\system$ equals the label $\labelfunc(\mathcal{T}(s))$ in the corresponding IMDP state.
Note that this implicitly requires all states $x \in \mathcal{T}^{-1}(s)$ in the same region to have the same label. %

\textbf{Actions.} 
We define one IMDP action $a_j$ for each IMDP state (except $s_\star$), such that $\Actions \coloneqq \{a_1, a_2 \ldots, a_v\}$. We partition the control space $\mathcal{U} \subseteq \mathbb{R}^m$ into $J \in \NN$ disjoint subsets $\bar{U}_1,\ldots,\bar{U}_J$, such that $\mathcal{U} = \bigcup_{j=1}^{J} \bar{U}_j$, with $\bar{U}_j \cap \bar{U}_{j'} = \emptyset$ for $j \neq j'.$
For each $\bar{U}_j$, we select a representative input, denoted by $\bar{u}_j \in \bar{U}_j$, collectively written as $\mathbf{\bar{u}} = \{\bar{u}_1, \bar{u}_2, \dots, \bar{u}_J\}$. 

Actions in the IMDP correspond to taking inputs $u \in \mathbf{\bar{u}}$ in system $\system$.
However, every action $a \in \Actions$ is not directly related to a single input $u \in \mathbf{\bar{u}}$ (as is common in many abstraction methods).
Instead, the input for $a \in \Actions$ also depends on the continuous state $x \in \mathcal{X}$.
This intuition is formalised by defining an \emph{interface} for every action.

\begin{definition}[Interface function]
    \label{def:interface}
    Every pair $(s,a) \in \States \times \Actions$ such that $a \in \Actions(s)$ is associated with an \emph{interface function} $I_{s,a}: \mathcal{T}^{-1}(s) \to \mathbf{\bar{u}} \subset \mathcal{U}$ that maps every continuous state $x \in \mathcal{T}^{-1}(s)$ to the control input $I_{s,a}(x) \in \mathbf{\bar{u}}$ associated with performing action $a$ in state $x$.
\end{definition}

Thus, a crucial step in our approach is to define the enabled actions $\Actions(s)$ for all $s \in \States$ and the interface functions $I_{s,a}$.
We, however, postpone this important step to \cref{sec:data}, where we will learn these elements from the dataset $\mathbf{D}_N$.

\textbf{Transitions.} 
For an IMDP state $s \in \States$ and a corresponding continuous state $x \in \mathcal{T}^{-1}(s)$, let $a \in \Actions(s)$ be an action executed in the IMDP.
Executing this action $a$ corresponds to executing the control input $u = I_{s,a}(x)$ in system $\system$ given by the interface function.
The probability that the next continuous state is contained in a fixed compact set $X' \subset \mathcal{X}$ is
\[
\transprob(x, X', a) \coloneqq \Prob \left\{ \omega \in \Omega : f(x, I_{\mathcal{T}(x),a}(x), \noise(\omega)) \in X' \right\},
\]
where the dependency on $I_{\mathcal{T}(x),a}$ is implicit.
We then define each interval of the IMDP transition function $\transfuncImdp(s,a)$ as
\begin{subequations}\label{eq:transition_function}
\begin{align}
\plow(s,a,s')&:=\min_{x \in \mathcal{T}^{-1}(s)} \transprob \big(x, \mathcal{T}^{-1}(s'), a \big),\\
\pupp(s,a,s')&:=\max_{x \in \mathcal{T}^{-1}(s)} \transprob \big(x, \mathcal{T}^{-1}(s'), a \big).
\end{align}
\end{subequations}

In \cref{sec:data}, we present our data-driven approach to learn $\plow$ and $\pupp$ for all state-action pairs based on the dataset $\mathbf{D}_N$.

\textbf{Correctness.} 
We show that any scheduler $\scheduler$ for the abstract IMDP can be \emph{refined} into a policy $\policy$ for the concrete system $\system$.
Importantly, if the IMDP is label-preserving, the probability that the IMDP satisfies a specification $\varphi$ (under $\scheduler)$ is a \emph{lower bound} on the satisfaction probability in the concrete system (under $\policy$).
A similar result for reach-avoid specifications has been proven in~\cite[Thm.~5.23]{ThesisThom2025}, which we generalise to the broader set of  specifications we consider in this paper.

\begin{theorem}[Policy synthesis]    \label{thm:existence_controller_main}
    Let $\imdp$ be the IMDP abstraction for system $\system$, let $\varphi \subseteq (2^\mathcal{Y})^\mathbb{N}$ be a specification, and let $\scheduler \in \schedulerSpace^\imdp$ be an IMDP scheduler.
    It holds that
    \begin{equation}
        \label{eq:existence_controller_main}
        \satprob_\pi^\system(\varphi) \, \geq \, \min_{P \in \transfuncImdp} \satprob^\imdp_{\scheduler,P}\{ (L(s_0), L(s_1), \ldots) \in \varphi' \},
    \end{equation}
    where the policy $\policy$ is constrained to the interface functions, i.e., for all $x_k \in \mathcal{X}$ and $k \in \NN$,
    \[
    \policy_k(x_k) \in I_{s,a}(x_k), \,\, s = \mathcal{T}(x_k), \, a = \scheduler_k(s).
    \]
\end{theorem}
\smallskip

\textit{Proof sketch.}
The proof follows by generalising the result in~\cite[Thm.~5.23]{ThesisThom2025} to the specifications we consider in this paper.
Choosing the policy $\policy$ of the system $\system$ to be contained in the interface functions $I_{s,a}$ implies there exists\footnote{More specifically, this step relies on establishing a \emph{probabilistic alternating simulation relation} from the IMDP $\imdp$ to the dynamical system $\system$. Due to space limitations, we refer to \cite{ThesisThom2025} for details about these relations.} an adversary $\tilde{P} \in \transfuncImdp$ for which system $\system$ and the IMDP $\imdp$ have \emph{equal distributions over the traces} $(2^\mathcal{Y})^\mathbb{N}$.
This, in turn, implies there exists $\tilde{P} \in \transfuncImdp$ such that the inequality in \cref{eq:existence_controller_main} holds with equality.
Hence, minimising over $P \in \transfuncImdp$ yields the inequality in \cref{eq:existence_controller_main}, which concludes the proof.

\section{Data-Driven IMDP Construction}
\label{sec:data}

We use the dataset $\mathbf{D}_N$ to learn the enabled IMDP actions $\Actions(s)$, interface functions $I_{s,a}$, and distributions $\transfuncImdp(s, a)$ defined by \cref{eq:transition_function}.
Recall $\mathbf{D}_N = ( x_i, u_i, x_i^+ )_{i=1}^N$ only consists of noisy samples, while the dynamics $f$ and the underlying noise realisations are unknown.
For every state $s \in \States$, we take the following steps:
\begin{enumerate}
    \item Determine the enabled actions $\Actions(s)$ in each abstract state $s$ based on the dataset $\mathbf{D}_N$. 
    \item Use the samples $(x,u,x^+) \in \mathbf{D}_N$ and the Lipschitz constants $L_X$ and $L_U$ to overapproximate the forward reachable sets needed for step~(3).
    \item Use the Clopper-Pearson interval~\cite{clopper1934use} to define bounds on the IMDP's intervals using the reachable sets for every control input $\bar{u}_j \in \mathbf{\bar{u}}$ and next state $s' \in \States$.
    \item Define interface functions for every action $a \in \Actions(s)$ using the probability intervals computed in step~(3).
\end{enumerate}
Moreover, we use a refinement scheme that subdivides each partition element into smaller cells, called voxels. 
These voxels lead to smaller reachable sets and, thus, tighter~intervals.

Next, we detail each step of our algorithm.
We make any dependency on the interface function $I_{s,a}$ implicit for brevity.

\subsection{Enabled actions}
\label{subsec:actions}
We enable the abstract action $a_j$ in the abstract state $s_i$, i.e., $a_j \in \Actions(s_i)$, if the fraction of observed transitions from $x \in \mathcal{T}^{-1}(s_i)$ to $\hat x \in \mathcal{T}^{-1}(s_j)$ exceeds a threshold $T$, i.e., if
$$
\frac{\left| \{ (x, u, x^+) \in \mathbf{D}_N : x \in \mathcal{T}^{-1}(s_i),\; x^+ \in \mathcal{T}^{-1}(s_j) \} \right|}{\left| \{ (x, u, x^+) \in \mathbf{D}_N : x \in \mathcal{T}^{-1}(s_i) \} \right|} \geq T.
$$

The threshold $T$ is a hyperparameter used to adjust the set of enabled abstract actions $\Actions(s)$ for each $s \in \system$.

\subsection{Data-driven probability intervals}
\label{subsec:PAC}
The forward reachable set for a set $X \subset \mathbb{R}^n$, action $a$, and noise value $
\omega$ is defined as $F(X, a, \omega) = \big\{ f(x, u, \noise(\omega)) : x \in X, \, u = I_{\mathcal{T}(x),a}(x) \big\}$.
Thus, $F(\mathcal{T}^{-1}(s), a, \omega)$ is the set of reachable states from any state $x \in \mathcal{T}^{-1}(s)$, the control input $u = I_{\mathcal{T}(x),a}(x)$, and the noise realization $\noise(\omega)$.
Inspired by~\cite[Sect.~7.2]{ThesisThom2025}, observe that the following bounds hold:
\begin{align*}
    \plow(s,a,s') &\geq \mathbb{P}
    \left\{ 
    \omega \in \Omega : F(\mathcal{T}^{-1}(s), a, \omega) \subseteq \mathcal{T}^{-1}(s')
    \right\},
    \\
    \pupp(s,a,s') &\leq \mathbb{P}
    \left\{ 
    \omega \in \Omega : F(\mathcal{T}^{-1}(s), a, \omega) \cap \mathcal{T}^{-1}(s') \neq \emptyset
    \right\}.
\end{align*}

As the noise distribution is unknown, we define a sample-based version of these bounds.
Let $\omega_1,\ldots,\omega_M$ be $M \in \NN$ noise samples.
Similar to~\cite{l4dc}, we define PAC bounds on $\check{P}(s, a, s')$ and $\hat{P}(s, a, s')$ using the \emph{Clopper-Pearson interval}~\cite{clopper1934use} with $M$ samples of the noise $\omega_1,\ldots,\omega_M$.

\begin{theorem}[Clopper-Pearson interval]
    \label{thm:PAC_interval}
    Fix $s, s' \in \States$ and $a \in \Actions$, and let $I_{s,a}$ be the corresponding interface function.
    Let $(\omega_1,\ldots,\omega_M) \in \Omega^M$ be $M \in \NN$ i.i.d. noise samples.
    For brevity, let $F_i = F(\mathcal{T}^{-1}(s), a, \omega_i)$ and define
    \begin{align*}
        \check{M} &= \left| \left\{ i=1,\ldots,M : F_i \subseteq \mathcal{T}^{-1}(s') \right\} \right|,
        \\
        \hat{M} &= \left| \left\{ i=1,\ldots,M : F_i \cap \mathcal{T}^{-1}(s') \neq \emptyset \right\} \right|.
    \end{align*}
    Then, for any $\beta \in (0,1)$, it holds that
    \begin{align}
        \label{eq:thm_interval}
        \Prob^M \Big\{ 
        (\omega_{1},\ldots,\omega_{M}) \in \Omega^M
        :
        {}&{} \check{P}_\text{lb} \leq \plow(s,a,s') \,\, \wedge \\ {}&{} \pupp(s,a,s') \leq \hat{P}_\text{ub}
        \Big\} \geq 1-\beta,
        \nonumber
    \end{align}
    where $\check{P}_\text{lb} = 0$ if $\check{M} = 0$, and otherwise, $\check{P}_\text{lb}$ is the solution~to
    \begin{equation}
        \label{eq:thm_lower_bound}
        \frac{\beta}{2} = \sum\nolimits_{\ell=\check{M}}^{M} \binom{M}{\ell} \cdot (\check{P}_\text{lb})^{\ell} \cdot (1-\check{P}_\text{lb})^{M-\ell},
    \end{equation}
    and $\hat{P}_\text{ub} = 1$ if $\hat{M} = M$, and otherwise, $\hat{P}_\text{ub}$ is the solution~to
    \begin{equation}
        \label{eq:thm_upper_bound}
        \frac{\beta}{2} = \sum\nolimits_{\ell=0}^{\hat{M}} \binom{M}{\ell} \cdot (\hat{P}_\text{ub})^\ell \cdot (1-\hat{P}_\text{ub})^{M-\ell}.
    \end{equation}
\end{theorem}

\begin{proof}
    Observe that $\check{M}$ and $\hat{M}$ are samples from binomial distributions\footnote{We write $\Bin(n,p)$ to denote a binomial distribution with $n \in \NN$ experiments and success probability $p \in [0,1]$.} with success probabilities $\check{p}(s,a,s')$ and $\hat{p}(s,a,s')$, i.e., 
    $
    \check{M} \sim \Bin(M,\check{p}(s,a,s')) 
    \,\text{and}\,
    \hat{M} \sim \Bin(M,\hat{p}(s,a,s')).
    $
    Applying the Clopper-Pearson interval~\cite{clopper1934use} to bound the success probabilities yields
    $
        \Prob^N \big\{ \check{P}_\text{lb} \leq \check{p}(s,a,s') \big\} \geq 1-\frac{\beta}{2}$, and 
        $\Prob^N \big\{ \hat{p}(s,a,s') \leq \hat{P}_\text{ub} \big\} \geq 1-\frac{\beta}{2}.
    $
    Combined through the union bound, we obtain \cref{eq:thm_interval}.
\end{proof}

\subsection{Data-driven reachable sets}
\label{subsec:FRS}
In practice, we do not have the noise samples and reachable sets $F(\cdot)$ needed in \cref{thm:PAC_interval}.
Still, we can overapproximate $F(\cdot)$ using the samples in the dataset $\mathbf{D}_N = ( x_i, u_i, x_i^+ )_{i=1}^N$.
With abuse of notation, let $\mathbf{D}_N(X, U) \subset \mathbf{D}_N$ be the subset of samples $(x, u, x^+)$ for which $x \in X$ and $u \in U$. Thus, $\mathbf{D}_N(\mathcal{T}^{-1}(s), \bar{U}_j)$ are the samples for which $x$ is contained in the region $\mathcal{T}^{-1}(s)$ associated with IMDP state $s$, and $u$ in the discrete control set $\bar{U}_j$ with representative point $\bar{u}_j$.
For a fixed sample $(x,u,x^+) \in \mathbf{D}_N(\mathcal{T}^{-1}(s), \bar{U}_j)$, define\footnote{Recall from \cref{assumption} that Lipschitz constants $L_X$, $L_U$ are known.}
\begin{align*}
\widehat{F}(x, u, x^+) \coloneqq \big\{ 
    \tilde{x} \in \mathcal{X} : {}&{} \|\tilde{x} - x^+\| \leq L_U \cdot \|u - \bar{u}_j\| \\
    &\!\!\!{} + L_X \cdot \max_{\hat x \in \mathcal{T}^{-1}(s)} \|x - \hat x\| 
\big\},
\end{align*}
where, implicitly, the next state $x^+$ is defined by an (unknown) noise realization $\noise(\omega)$.
By definition of the Lipschitz constants in \cref{eq:lipschitz}, $\widehat{F}(x, u, x^+)$ overapproximates the forward reachable set $F(\mathcal{T}^{-1}(s), a, \omega)$ under the noise realization $\noise(\omega)$ and the interface function defined by $I_{s,a} = \bar{u}_j$, i.e.,
\begin{equation}
    \label{eq:overapproxiate_FRS}
    F(\mathcal{T}^{-1}(s),a,\omega) \subseteq \widehat{F}(x, u, x^+),
\end{equation}
Therefore, we directly obtain the following key result.
\begin{corollary}
    \label{eq:cor:PAC_intervals}
    Fix $j \in \{1,2,\ldots,J\}$.
    Consider again \cref{thm:PAC_interval}, but with $M \coloneqq | \mathbf{D}_N(\mathcal{T}^{-1}(s), \bar{U}_j) |$ and with $F_i \coloneqq \widehat{F}(x_i, u_i, x^+_i)$ for all $i=1,2,\ldots,M$.
    Then, for any $\beta \in (0,1)$, the PAC bounds in \cref{eq:thm_interval} hold.
\end{corollary}

\begin{proof}
    Due to \cref{eq:overapproxiate_FRS}, the value of $\check{M}$ for $\widehat{F}(\cdot)$ is lower than for $F(\cdot)$, leading to a lower value of $\check{P}_\text{lb}$.
    Conversely, the value of $\hat{M}$ for $\widehat{F}(\cdot)$ is higher than for $F(\cdot)$, leading to a higher $\check{P}_\text{ub}$.
    Thus, the claim follows.
\end{proof}

\subsection{Interface functions}
For each IMDP state $s_i$, we first determine the enabled actions $\Actions(s_i)$ using the dataset $\mathbf{D}_N$, as described in~\cref{subsec:actions}. 
Then, as explained in~\cref{subsec:FRS}, we compute overapproximations $\widehat{F}(x, u, x^+)$ of the reachable sets for the samples in $\mathbf{D}_N(\mathcal{T}^{-1}(s_i), \bar{U}_j)$ for all $1\leq j \leq J$. For each interface function $I_{s_i,a} = \bar{u}_j$, we then compute probability intervals $[\check{P}_\text{lb}(s_i,a,s'), \hat{P}_\text{lb}(s_i,a,s')]$ for all successor states $s' \in S$.

Each action $a_k \in \Actions(s_i)$ guides the system towards a particular state $s_k$. We define the interface $I_{s_i,a_k}$ to pick the control $\bar{u} \in \mathbf{\bar{u}}$ maximising the probability of reaching $s_k$:
\begin{align*}
    I_{s_i,a_k} = \argmax_{\bar{u} \in \mathbf{\bar{u}}} 
    \big[{}&{}\max(\check{P}_\text{lb}(s_i,a^{(\bar{u})},s_k), \quad 1-
    \\
    &\sum_{s'\in S\setminus\{s_k\}}\hat{P}_\text{lb}(s_i,a^{(\bar{u})},s')), \ \ I_{s_i,a^{(\bar{u})}} = \bar{u}\big].
\end{align*}

\subsection{Abstraction refinement}
\label{subsec:voxels}

To make the abstraction less conservative, we divide each region $R_i$ into $\varphi_i$ smaller, non-overlapping subregions $\Phi_i=\{\phi_1, \phi_2, \ldots, \phi_{\varphi_i}\}$ called voxels, such that $R_i = \bigcup_{\ell=1}^{\varphi_i} \phi_\ell$. Intuitively, these voxels allow for policies that combine different inputs depending on the continuous state within each region $R_i$. The approach outlined so far can be seen as a special case with a single voxel per IMDP state.

For each voxel $\phi_\ell \subseteq s_i$, we introduce a virtual state $s_{i}^{\phi_\ell}$ such that $\mathcal{T}^{-1}(s_{i}^{\phi_\ell}) = \phi_\ell$.
These states $s_{i}^{\phi_\ell}$ are named \emph{virtual} because the abstraction does not explicitly represent each voxel as an IMDP state. %
Following steps (2-4) for each virtual states $s_i^{\phi_\ell}$, we compute transition probabilities $[\check{P}_\text{lb}(s_i^{\phi_\ell},a_k,s'), \hat{P}_\text{ub}(s_i^{\phi_\ell},a_k,s')]$ and interface functions $I_{s_i^{\phi_\ell}, a_k}$ for every $s_i,s' \in S, \phi_\ell \in \Phi_i$, and $a_k \in \Actions(s_i)$.
These intervals need to be combined to construct the IMDP with abstract state $s$. 
To preserve soundness, we conservatively aggregate these local intervals by taking their union and assigning the resulting interval as the transition probability from $s_i$ to $s'$ under action $a_k$ in the abstraction. Formally,
\begin{equation}
    \check{P}_\text{lb}(s_i,a_k,s') = \min_{1 \leq \ell \leq \varphi_i} \check{P}_\text{lb}(s_i^{\phi_\ell},a_k,s'),
\end{equation}
\begin{equation}
    \hat{P}_\text{ub}(s_i,a_k,s') = \max_{1 \leq \ell \leq \varphi_i} \hat{P}_\text{ub}(s_i^{\phi_\ell},a_k,s').
\end{equation}
We aggregate the local interface functions corresponding to each voxel $\phi_\ell$ to define a global one for state $s_i$:
\begin{equation}
    I_{s_i, a_k}(x) = \cup_{\ell=1}^{\varphi_i} I_{s_i^{\phi_\ell}, a_k}(x)\cdot \mathbf{1}_{\{x \in \phi_\ell\}},
\end{equation}
where $\mathbf{1}_{\{\cdot\}}$ is the indicator function. Since voxels are disjoint, exactly one indicator function is active for each $x \in \mathcal{T}^{-1}(s_i)$.

In general, there is a trade-off between computation time (increasing with the number of voxels) and the accuracy of the abstraction (which improves with the number of voxels). We investigate this trade-off empirically in the next section.

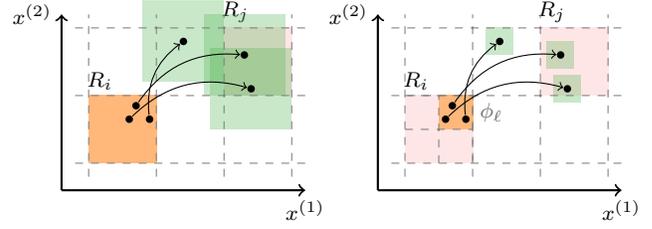
\begin{figure}[t!]
\centering
\begin{tikzpicture}[scale=0.9, font=\footnotesize]
    
    \fill[red!10] (1,1) rectangle (2,2) node [pos=0.5, yshift=0.4cm, xshift=-0.3cm, above, black] {$R_i$};
    \fill[orange!80, opacity=0.6] (1.0,1.0) rectangle (2.0,2.0);
    \fill[red!10] (3,2) rectangle (4,3) node [pos=0.5, yshift=0.4cm, xshift=-0.3cm, above, black] {$R_j$};
    
    \fill[green!40!gray, opacity=0.3] (2.7,2.0) rectangle (3.9,3.2);
    \fill[green!40!gray, opacity=0.3] (2.8,1.5) rectangle (4.0,2.7);
    \fill[green!40!gray, opacity=0.3] (1.8,2.2 ) rectangle (3.0,3.4);
    
    \draw[step=1cm,gray,dashed] (0.8,0.8) grid (4.2,3.2);
            
    \node (x1) [circle, fill=black, inner sep=1pt, minimum size=1pt]at (1.70,1.85) {};
    \node (x1plus) [circle, fill=black, inner sep=1pt, minimum size=1pt] at (3.3,2.6) {};
    \draw[->, thin, bend left] (x1) to (x1plus);
    
    \node (x2) [circle, fill=black, inner sep=1pt, minimum size=1pt]at (1.60,1.65) {};
    \node (x2plus) [circle, fill=black, inner sep=1pt, minimum size=1pt] at (3.4,2.1) {};
    \draw[->, thin, bend left] (x2) to (x2plus);

    \node (x3) [circle, fill=black, inner sep=1pt, minimum size=1pt]at (1.90,1.65) {};
    \node (x3plus) [circle, fill=black, inner sep=1pt, minimum size=1pt] at (2.4,2.8) {};
    \draw[->, thin, bend left] (x3) to (x3plus);

    \draw[thick,->] (0.6,0.6) -- (4.2,0.6) node[below] {$x^{(1)}$};
    \draw[thick,->] (0.6,0.6) -- (0.6,3.2) node[left] {$x^{(2)}$};
\end{tikzpicture}\hspace{-0.3cm}
\begin{tikzpicture}[scale=0.9, font=\footnotesize]
    \fill[red!10] (1,1) rectangle (2,2) node [pos=0.5, yshift=0.4cm, xshift=-0.3cm, above, black] {$R_i$};
    \fill[orange!80, opacity=0.6] (1.5,1.5) rectangle (2.0,2.0) node [pos=0.5, yshift=0cm, xshift=0.2cm, right, black] {$\phi_\ell$};
    \fill[red!10] (3,2) rectangle (4,3) node [pos=0.5, yshift=0.4cm, xshift=-0.3cm, above, black] {$R_j$};
    
    \fill[green!40!gray, opacity=0.3] (3.1,2.4) rectangle (3.5,2.8);
    \fill[green!40!gray, opacity=0.3] (3.2,1.9) rectangle (3.6,2.3);
    \fill[green!40!gray, opacity=0.3] (2.2,2.6) rectangle (2.6,3);
    
    \draw[step=1cm,gray,dashed] (0.8,0.8) grid (4.2,3.2);    \draw[step=0.5cm,gray,dashed] (1,1) grid (2,2);

    \node (x1) [circle, fill=black, inner sep=1pt, minimum size=1pt]at (1.70,1.85) {};
    \node (x1plus) [circle, fill=black, inner sep=1pt, minimum size=1pt] at (3.3,2.6) {};
    \draw[->, thin, bend left] (x1) to (x1plus);
    
    \node (x2) [circle, fill=black, inner sep=1pt, minimum size=1pt]at (1.60,1.65) {};
    \node (x2plus) [circle, fill=black, inner sep=1pt, minimum size=1pt] at (3.4,2.1) {};
    \draw[->, thin, bend left] (x2) to (x2plus);

    \node (x3) [circle, fill=black, inner sep=1pt, minimum size=1pt]at (1.90,1.65) {};
    \node (x3plus) [circle, fill=black, inner sep=1pt, minimum size=1pt] at (2.4,2.8) {};
    \draw[->, thin, bend left] (x3) to (x3plus);

    \draw[thick,->] (0.6,0.6) -- (4.2,0.6) node[below] {$x^{(1)}$};
    \draw[thick,->] (0.6,0.6) -- (0.6,3.2) node[left] {$x^{(2)}$};
\end{tikzpicture}%
%
\vspace{-3em}%
\caption{
Three samples $(x, u, x^+)$ with $x \in R_i$, and overapproximations of forward reachable sets $\widehat{F}(x, u, x^+ \mid s_i, \bar{u}_j)$ around each $x^+$, under the discrete control input $\bar{u}_j \in \bar U_j$.
In the left figure, one voxel per abstract state is used, resulting in large green boxes that form coarse overapproximations. In the right figure, four voxels (two per dimension) are used, leading to smaller and tighter overapproximations of the forward reachable sets. This allows for less conservative abstractions, as demonstrated in our experimental results.
}
\label{fig:samples}
\end{figure}

\section{Experimental Evaluation}
\label{sec:experiments}
We implemented our algorithm in Python and evaluated it on a variant of the car parking benchmark from~\cite{l4dc}. 
To compute~robust schedulers (as defined in \cref{eq:optimal_policy}), we use robust value iteration within the PRISM model checker~\cite{DBLP:conf/cav/KwiatkowskaNP11}. 
The experiments are run on a machine with two 3.3~GHz~cores and 128~GB of RAM. 
The benchmark from~\cite{l4dc} models the $(x,y)$ position of a car, 
where the inputs $u_k = [v_k, \theta_k]$, i.e., the velocity and steering angle, are constrained to $u_k \in \mathcal{U} = [-0.2, 0.2] \times [-\pi, \pi]$.
The dynamics evolve as
\begin{equation}
\begin{split}
    x_{k+1} &= x_k + 10 (v_k + \zeta_1)  \cos(\theta_k + \zeta_2) + \zeta_3,
    \\
    y_{k+1} &= y_k + 10 (v_k + \zeta_1) \sin(\theta_k + \zeta_2) + \zeta_4,
    \label{eq:experiment_dynamics}
\end{split}
\end{equation}
where the noise $\zeta \sim N(0, \sigma^2 \mathbf{I}_4)$ is normally distributed with $\sigma^2=0.01$. 
The multiplicative noise makes this benchmark more challenging to solve than only with additive noise.
We consider an unbounded reach-while-avoid specification with the goal region $\mathcal{X}_G = [10, 11] \times [10, 11]$ and unsafe region $\mathcal{X}_U = ( \RR^2 \setminus ([-10, 10] \times [-10, 10]) ) \cup ([3, 4] \times [0, 9]) \cup ([7, 9] \times [3, 12])$, creating the maze-like environment in~\cref{fig:traj}. %

\noindent\textbf{Dataset.}
We use $N = 2 \times 10^8$ noisy state-input-state samples $(x,u,x^+)$, where $x$ is sampled uniformly from $\mathcal{X}$, $u$ is sampled uniformly from a discretised set of $11 \times 11$ control inputs in the control space $\mathcal{U}$, and $x^+$ follows from \cref{eq:experiment_dynamics}.
Note that the noise realisations are not recorded in the dataset.

\noindent\textbf{Abstraction.} 
We use a uniform partition into $12 \times 12$ states, using $9$ voxels per state. The threshold $T=0.05$ is used to determine enabled actions.
The IMDP abstraction has $144$ states, $827$ actions, and $8\,199$ transitions. 
Generating the abstraction takes approximately $72$ minutes, and computing an optimal policy using PRISM less than one second. We use a confidence of $1 - \beta = 1 - \frac{0.05}{10\,000}$ to compute each probability interval as per \cref{thm:PAC_interval}. 
As the total number of transitions is below $10\,000$, the full abstraction is correct with a probability $\geq 0.95$.

\noindent\textbf{Results.}
In~\cref{fig:heatmaps} we plot the reach-while-avoid satisfaction probabilities obtained using $(2 \times 2)$ and $(3 \times 3)$ voxel grids per abstract state (as described in \cref{subsec:voxels}). 
We also tried using $(1\times 1)$ voxels, i.e., when no interface function refinement is used; however, that leads to vacuous satisfaction probabilities, indicating that the abstraction is too conservative.
We also plot 100 simulations of the closed-loop system over the $(3 \times 3)$ voxel abstraction in~\cref{fig:traj}, showing that our algorithm yields safe and robust control policies in practice.

\begin{figure}[t!]
\centering
\begin{subfigure}{0.48\linewidth}
    \centering
    \includegraphics[height=3.5cm]{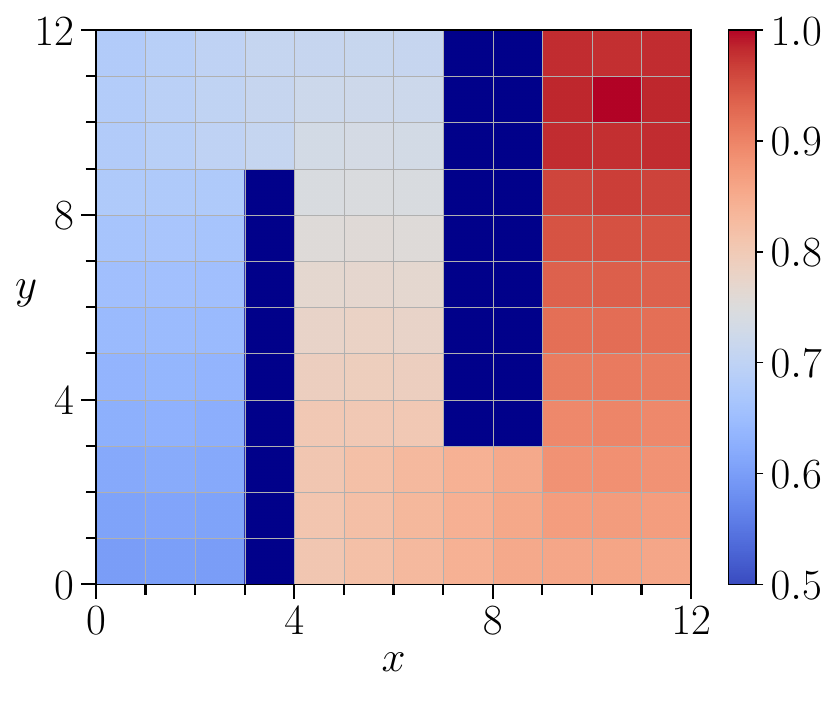}%
    \vspace{-0.6em}
    \caption{$3 \times 3$ voxels per state.}
\end{subfigure}%
\begin{subfigure}{0.48\linewidth}
    \centering
    \includegraphics[height=3.5cm]{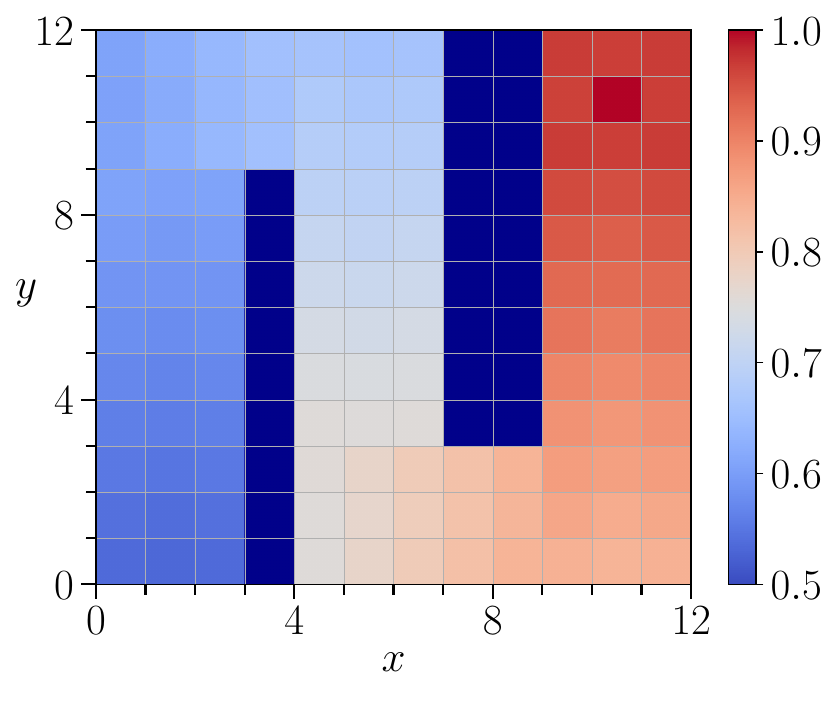}%
    \vspace{-0.6em}
    \caption{$2 \times 2$ voxels per state.}
\end{subfigure}
\caption{Satisfaction probability in different abstract initial states for (a) $3 \times 3$ and (b) $2 \times 2$ voxels per abstract state. 
}
\label{fig:heatmaps}
\vspace{0.3em}
\centering
\includegraphics[height=3.5cm]{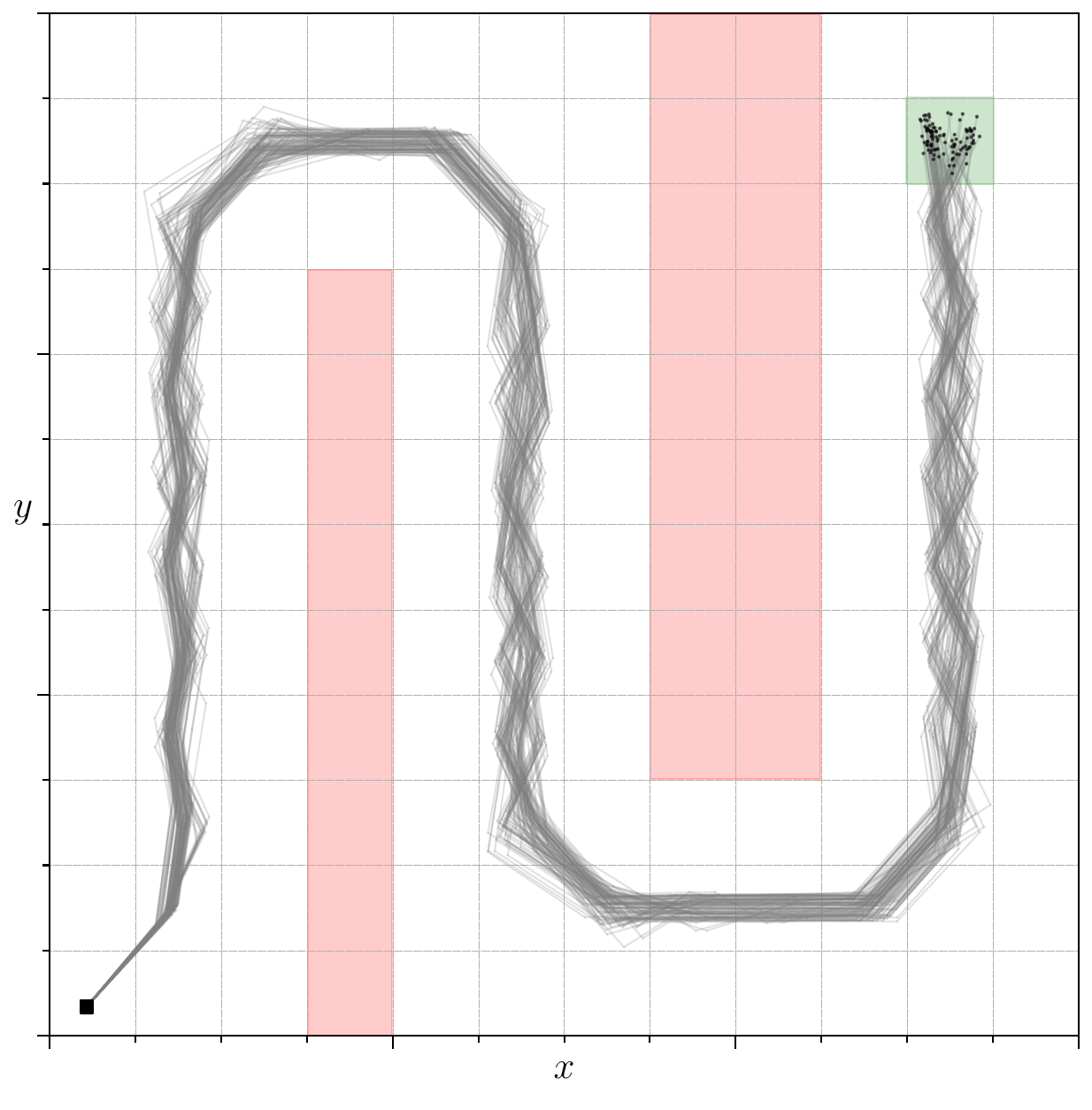}
\vspace{-0.6em}%
\caption{Simulation of 100 trajectories of the closed-loop system using the refined controller.}
\vspace{-1.75em}%
\label{fig:traj}
\end{figure}

\section{Conclusion}
\label{sec:conclusion}

We presented a data-driven algorithm for constructing abstractions of discrete-time stochastic systems with unknown dynamics.
As a distinctive and novel feature, our algorithm only requires noisy state-input-state observations.
Together with (an upper bound on) the Lipschitz constant of the dynamics, we used these noisy observations to learn an IMDP abstraction with PAC guarantees.
Our algorithm enables controller synthesis with guarantees, even when only a black-box simulation model of the system is available.

In this work, we laid the theoretical foundations for a data-driven abstraction framework. A direct implementation of the algorithm is computationally expensive as our numerical results show. Future work should focus on reducing computational complexity through parallelisation, compositional techniques, and learning-based methods.

\bibliographystyle{ieeetr}
\bibliography{references.bib}

\end{document}